\documentclass{llncs}    
\usepackage{fullpage}
\usepackage{amssymb}
\usepackage{xspace}
\usepackage[english]{babel}
\usepackage{amsmath}
\usepackage{latexsym}
\usepackage{textcomp}
\usepackage{epsfig}
\usepackage{subfigure}
\usepackage{color}
\usepackage{times}
\usepackage{enumerate}
\usepackage{tabularx}
\usepackage[labelfont=bf,size=small]{caption}
\usepackage[plainpages=false]{hyperref}
\usepackage[figure,table]{hypcap}
\usepackage{multirow}
\usepackage{rotating}
\usepackage{calc}
\usepackage[textsize=footnotesize]{todonotes}
\usepackage{graphics}
\usepackage{graphicx}

\hypersetup{pdfborder={0 0 0.4}}

\newcommand{\remove}[1]{}

\definecolor{blue}{rgb}{0.274,0.392,0.666}
\definecolor{red}{rgb}{0.627,0.117,0.156}

\title{L-Drawings of Directed Graphs}

\author{Patrizio Angelini\inst{1}, Giordano Da Lozzo\inst{2}, Marco Di Bartolomeo\inst{2},\\ Valentino Di Donato\inst{2}, Maurizio Patrignani\inst{2}, Vincenzo Roselli\inst{2}, Ioannis G. Tollis\inst{3}
}

\institute{
Wilhelm-Schickard-Institut f\"ur Informatik, Universit\"at T\"ubingen, Germany \and
Department of Engineering, Roma Tre University, Italy \and
University of Crete and Institute of Computer Science-FORTH, Greece}

\begin{document}

\pagestyle{plain}
\setcounter{page}{1}

\maketitle

\begin{abstract}
We introduce \emph{L-drawings}, a novel paradigm for representing directed graphs aiming at combining the readability features of orthogonal drawings with the expressive power of matrix representations. In an L-drawing, vertices have exclusive $x$- and $y$-coordinates and edges consist of two segments, one exiting the source vertically and one entering the destination horizontally.

We study the problem of computing L-drawings using minimum ink. We prove its NP-completeness and provide a heuristics based on a polynomial-time algorithm that adds a vertex to a drawing using the minimum additional ink. We performed an experimental analysis of the heuristics which confirms its effectiveness.
\end{abstract}

%%%%%%%%%
%%%%%%%%%
%%%%%%%%%
%%%%%%%%%
%%%%%%%%%
\section{Introduction}\label{se:intro}

%%%%%%%
%     %  
%     %  
%%%%%%%
\begin{figure}[tb]
\begin{center}
\begin{tabular}{c @{\hspace{1em}} c  c  }
   \includegraphics[width=3.9cm]{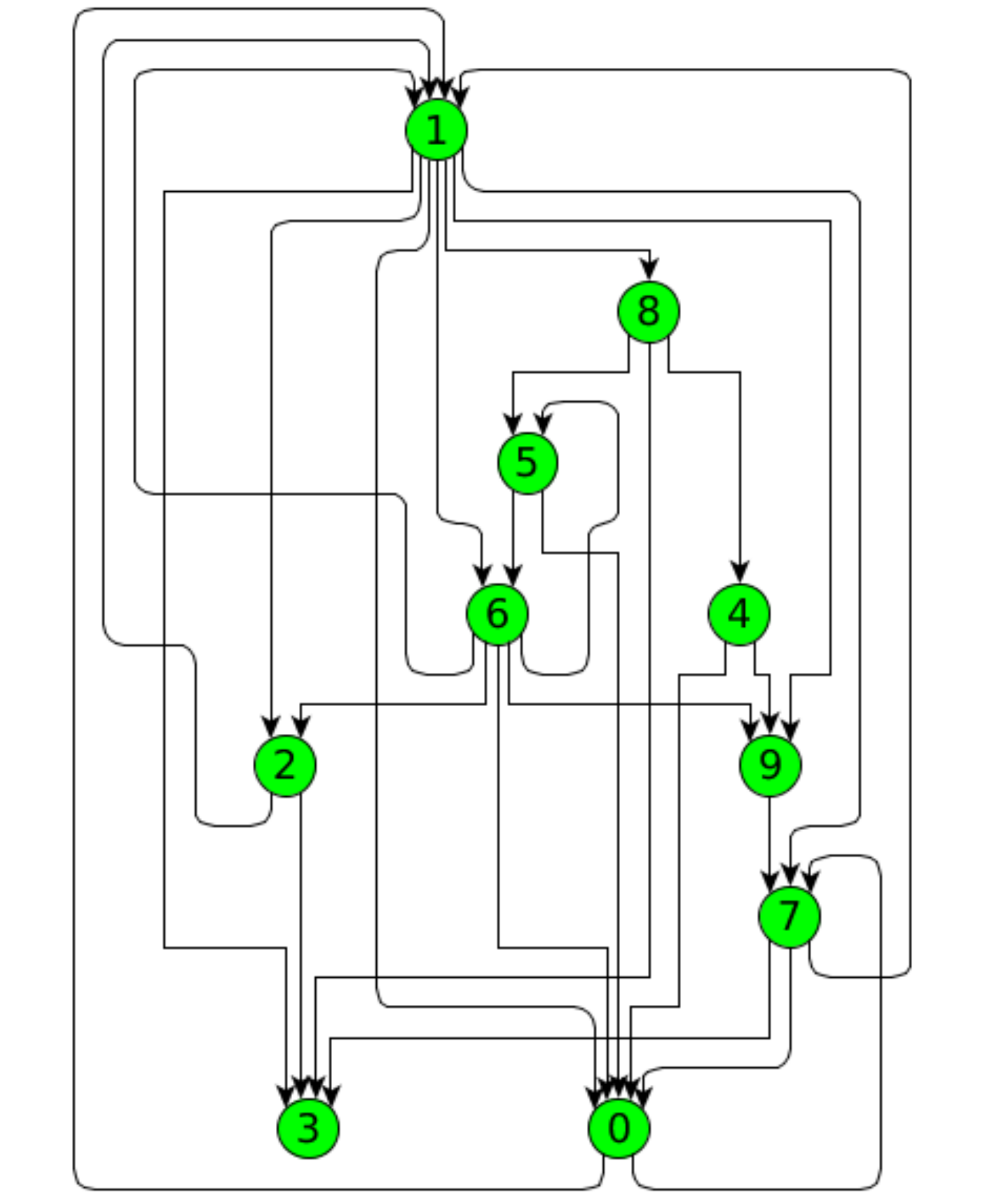} & 
   \includegraphics[width=4cm]{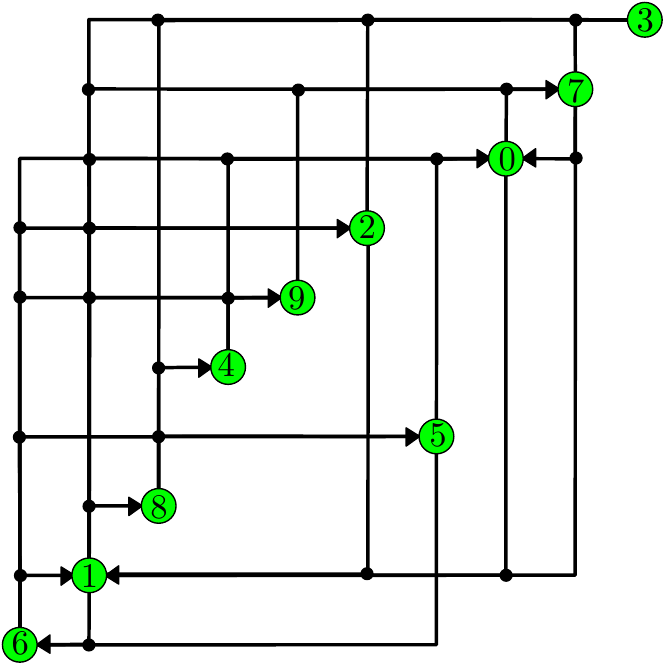} & 
   \includegraphics[width=4cm]{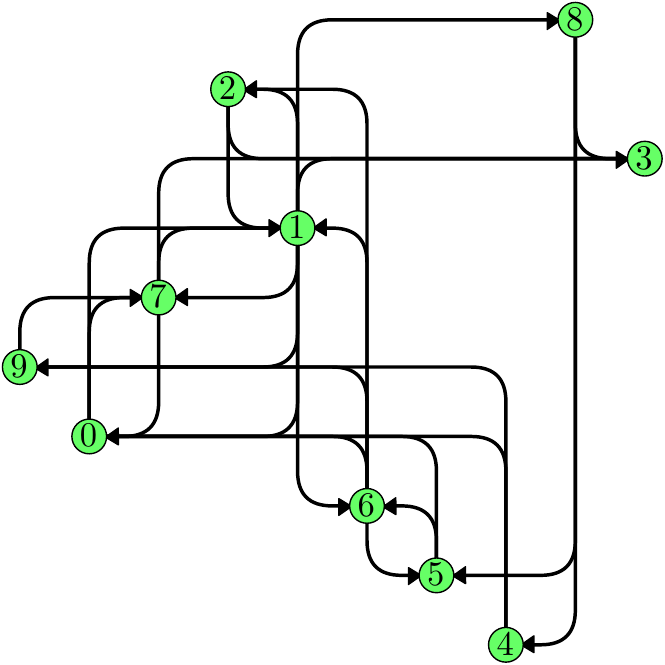} \\
   (a) & (b) & (c)
\end{tabular}
    \caption{(a) A hierarchical drawing with ``backloop routing'' produced by yEd~\cite{yed}, (b) an OOD, and a (c) minimum-ink L-drawing of the same connected random directed graph. 
    }\label{fi:opening-example}
    \end{center}
\end{figure}

Drawing directed graphs\remove{ while conveying their hierarchical structure} is a challenging goal to which a vast literature has been dedicated~\cite{hn-hda-13,stt-mvuhs-81}. 
In fact, most of the theoretical and applicative tasks concerning these graphs turned out to be difficult. To give a few examples, even for planar and acyclic graphs, it is hard to decide whether they admit an upward planar drawing~\cite{gt-ccurp-01}; if a directed graph contains directed cycles, it is hard to reverse the minimum number of edges to make it acyclic~\cite{g--66,hk-gafsp-03}, which is the first step of the renown Sugiyama approach~\cite{stt-mvuhs-81}. From a practical perspective, the more directed cycles it has, the less a hierarchical drawing of it becomes meaningful, strongly reducing the possibility of obtaining a clear and unambiguous representation.

In this paper we introduce a novel drawing paradigm specifically conceived for directed graphs, which combines orthogonal drawings with matrix representations. Na\-me\-ly, we call \emph{L-drawing} a drawing where each vertex has exclusive $x$- and $y$-coordinates and each directed edge has two orthogonal segments, one leaving the source vertically and one entering the destination horizontally. Edges are allowed both to overlap and to intersect. Graphically, the joint between the horizontal and the vertical segment of an edge is drawn as a small circular arc, allowing the user to easily identify the edges even in the presence of overlaps and intersections. We remark that L-drawings are strictly related to the popular {\em confluent drawing} style~\cite{degm-cd-05}, which also leverages partially collinear edges and smoothened bends to reduce the visual complexity of the representation.

An example of L-drawing is in Fig.~\ref{fi:opening-example}(c); further examples can be found in Fig~\ref{fi:examples-2}.

This paradigm is inspired by the \emph{overloaded orthogonal drawings}~\cite{kt-ood-11,kt-davlg-12} of directed acyclic graphs, in which vertices have exclusive $x$- and $y$-coordinates and the edges consist of two segments, one leaving the source from the top and one entering the destination from the left. For graphs that are not acyclic, a minimal set of edges is selected to be drawn backward, leaving the source from the bottom and entering the destination from the right. Edges are hence always drawn with a single bend turning clockwise. As long as the graph has few directed cycles, this model results extremely effective, as also testified by user studies~\cite{dmpt-hvdgu-14}. L-drawings can in fact be seen as a generalization of this model to graphs that may contain many directed cycles, so that edges are allowed both to turn clockwise and counterclockwise. Note that, instead of using small circular arcs, ambiguities are solved in overloaded orthogonal drawings by placing a small dot on each overlapped bend (see Fig.~\ref{fi:opening-example}(b)). 

The relationship of L-drawings with orthogonal drawings is immediate and the benefits are immediate as well, since orthogonal drawings are widely recognized as one of the most readable drawing standards, ensuring a clear readability even in the presence of crossings~\cite{dett-gd-99,hhe-eca-08}.
We remark that a representation very similar to L-drawings was used in~\cite{bk-aesig-97} as an intermediate step to compute orthogonal drawings of high-degree graphs in the Kandinsky model. However, the main purpose of~\cite{bk-aesig-97} was to balance edges on the four sides of each vertex, so to reduce the area of the orthogonal drawing obtained once the vertices are expanded into rectangular boxes.

The relationship with matrix representations, and the benefits deriving from it, are also somehow evident. User studies suggest that matrix representations are extremely well suited for many simple tasks, but their performances dramatically decrease when it is requested to follow paths in the graph~\cite{dmpt-hvdgu-14,gfc-rgunm-05}. This is due to the fact that in this representation each vertex has two labels, one for its row and one for its column. Traversing a directed edge consists of moving along the row of the source vertex until the column of the destination vertex is reached. Traversing a directed path, instead, forces the user to repeatedly jump from the column of the vertex that is reached to the row of the same vertex when it is left. L-drawings overcome this limitation by moving the labels inside the matrix. The matrix itself is symbolically represented by the edges, that identify the portions of the rows and columns that have to be followed to connect adjacent vertices. A previous attempt to combine node-link and matrix representations was presented in~\cite{hfm-ntrix-07}, which introduced the NodeTrix visualization tool.

L-drawings have several strong points: (i) they always exist and are easy to compute; in fact any placement of the vertices such that no two vertices share the same horizontal or vertical grid line yields a valid L-drawing (the placement of the vertices uniquely determines the routing of the edges); (ii) they are not ambiguous, even for very dense graphs; (iii) they are particularly suited for interactive graph drawing, since vertices and edges can be easily added or removed preserving the user's mental map.

Since L-drawings always exist, we are interested in producing readable ones. One of the most desirable features of a graph drawing, especially when the graph is large, is that of having a small size. The classical notion of size of a drawing, namely the area of its bounding box, does not make much sense in this case, due to the requirement of using different $x$- and $y$-coordinates. We hence study the problem of minimizing the \emph{ink} of the drawing, which is computed as the sum of the lengths of vertical and horizontal segments, where overlapping portions are counted only once.

We prove in Section~\ref{se:complexity} that this problem is NP-complete. Motivated by this, we describe in Section~\ref{se:online} an incremental heuristics, based on adding vertices one at a time using the minimum additional ink. This heuristics is experimentally evaluated in Section~\ref{se:experiments} against the optimal solution (when it was possible to compute one), against overloaded orthogonal drawings, and against a random placement of the vertices. We give definitions in Section~\ref{se:preliminaries} and conclude in Section~\ref{se:conclusions} suggesting future lines of research.

%%%%%%%%%
%%%%%%%%%
%%%%%%%%%
%%%%%%%%%
%%%%%%%%%
\section{Preliminaries}\label{se:preliminaries}

In this paper we consider graphs $G=(V,E)$ that are directed. An edge $(u,v)$ is an \emph{outgoing} edge of $u$ and an \emph{incoming} edge of $v$. We allow $G$ to contain both $(u,v)$ and $(v,u)$, but only a single copy of them; further, we do not allow loops $(u,u)$.  

In an \emph{L-drawing} $\Gamma$ of $G$ each vertex $v \in V$ is assigned an exclusive integer $x$-coordinate $x_v$ and $y$-coordinate $y_v$, and each edge $(u,v)$ is drawn as a 1-bend polyline composed of a vertical segment incident to $u$ and a horizontal segment incident to $v$. Note that, edges may cross and partially overlap. We resolve the ambiguity among crossings and bends by replacing each bend with a small rounded junction (see Fig.~\ref{fi:opening-example}(c)).

The \emph{ink} $ink(\Gamma)$ of an L-drawing $\Gamma$ is the sum of the lengths of vertical and horizontal segments, where overlapping portions are counted only once. Since rounded junctions have all equal size, they are not taken into account when measuring ink.

We are interested in producing L-drawings of minimum cost. Both if the cost is computed in terms of area or in terms of ink, it is immediate that a drawing of minimum cost uses contiguous values for the coordinates of the vertices. Also, since area and ink do not change up to a translation of the whole drawing, in the rest of the paper we assume to use integer $x$- and $y$-coordinates in the range $[1\dots n]$.   
With the above assumptions, given a graph $G=(V,E)$, an L-drawing can be immediately obtained by choosing any two orderings $\pi_x$ and $\pi_y$ for the vertices in $V$, where $\pi_x$ determines $x$-coordinates and $\pi_y$ determines $y$-coordinates. We denote such a drawing by $\Gamma(\pi_x,\pi_y)$, and its ink by $ink(\pi_x,\pi_y)$. For any two orderings $\pi_x$ and $\pi_y$, drawing $\Gamma(\pi_x, \pi_y)$ has area $n \times n$, where $n = |V|$. 
Hence, we focus on the problem of computing L-drawings with minimum ink. The corresponding decision problem is formally defined as follows.

\begin{quote}
{\bf Problem:} 
\begin{minipage}[t]{.90\linewidth} 
{\sc Minimum-Ink-L-Drawing (MILD)}
\end{minipage}\par\vspace{2pt}
{\bf Instance:} 
\begin{minipage}[t]{.90\linewidth} 
A directed graph $G=(V,E)$ and an integer $k$.
\end{minipage}\par\vspace{2pt}
{\bf Question:} 
\begin{minipage}[t]{.90\linewidth} 
Does $G$ admit an L-drawing $\Gamma$ such that $ink(\Gamma) \leq k$?
\end{minipage}
\end{quote}

Let $\Gamma$ be an L-drawing of $G$ and let $ink_x(\Gamma)$ ($ink_y(\Gamma)$, respectively) be the amount of ink used for horizontal (vertical, respectively) segments. Obviously, $ink(\Gamma) = ink_x(\Gamma)+ink_y(\Gamma)$. In the following lemma we prove that $ink_x(\Gamma)$ ($ink_y(\Gamma)$, respectively) only depends on the horizontal (vertical, respectively) permutation of the vertices in $\Gamma$, which makes it possible to search for two optimal permutations \mbox{independently.}

\begin{lemma}\label{le:separazione}
Let $G$ be a graph and let $\pi_x$ be any permutation of its vertices. For any two permutations $\pi'_y$ and $\pi''_y$ we have that $ink_x(\pi_x,\pi'_y) = ink_x(\pi_x,\pi''_y)$. Symmetrically, $ink_y(\pi'_x,\pi_y) = ink_y(\pi''_x,\pi_y)$ for any two permutations $\pi'_x$ and $\pi''_x$. 
\end{lemma}
\begin{proof}
Each edge $(u,v)$ is composed of two segments, one incident to the source vertex $u$ and one incident to the target vertex $v$. Hence, if we consider for each vertex only the segments incident to it, then all the segments of the drawing are eventually accounted for. Since overlaps are counted only once, $ink(\Gamma)$ is the sum, for every vertex, of the longest segments exiting it in the four directions North, East, South, and West. Thus, $ink_x(\Gamma)$ is the sum, for every vertex, of the longest segments exiting it in the directions East and West, while $ink_y(\Gamma)$ is the sum of the longest segments exiting it along North and South. Hence, $ink_x(\Gamma)$ only depends on $\pi_x$ and $ink_y(\Gamma)$ only depends on $\pi_y$. \qed
\end{proof} 

The \emph{complete graph $K_n$} is the directed graph $G=(V,E)$, where $|V| = n$ and for all ordered pairs $u,v \in V$, $u \neq v$, we have $(u,v) \in E$. In the following lemma we prove that any placement of the vertices of $K_n$ on the $n \times n$ grid yields an L-drawing whose edges use all the segments of such a grid.

\begin{lemma}\label{le:complete-optimal}
Any L-drawing $\Gamma$ of $K_n$ on the $n \times n$ grid uses $2n(n-1)$ ink. 
\end{lemma}
\begin{proof}
Consider the vertices $v_N$ and $v_S$ on the topmost and bottommost row of $\Gamma$, respectively, and the vertices $v_W$ and $v_E$ on the leftmost and rightmost column of $\Gamma$, respectively; note that $v_N \neq v_S$, $v_W \neq v_E$, and $0 \leq |\{v_N,v_S\} \cap \{v_E,v_W\}| \leq 2$.
For each column $x$, consider the vertex $v \in K_n$ lying in column $x$. Then, the vertical segments of edges $(v,v_N)$ and $(v,v_S)$ span the whole column $x$. Analogously, for each row $y$, consider the vertex $v \in K_n$ lying in row $y$. Then, the horizontal segments of edges $(v_W,v)$ and $(v_E,v)$ span the whole row $y$.
Hence, all the rows and columns of the $n \times n$ grid are spanned by at least one edge, and the statement trivially follows.
\qed
\end{proof}

%%%%%%%
%     %  
%     %  
%%%%%%%
\begin{figure}[htb]
\begin{center}
\begin{tabular}{c @{\hspace{1em}} c @{\hspace{1em}} c  }
   \includegraphics[width=4.3cm]{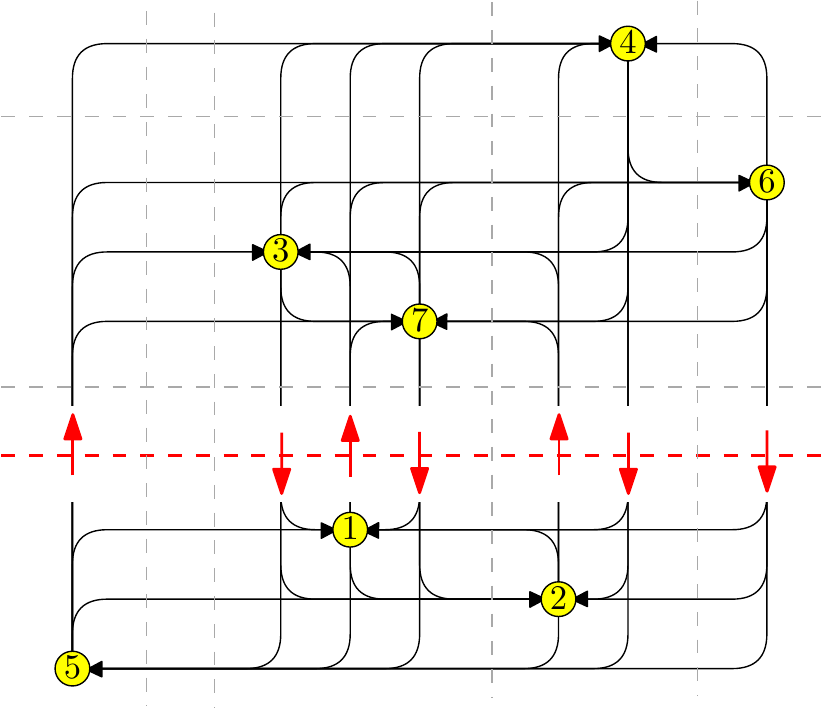} &
   \includegraphics[width=4cm]{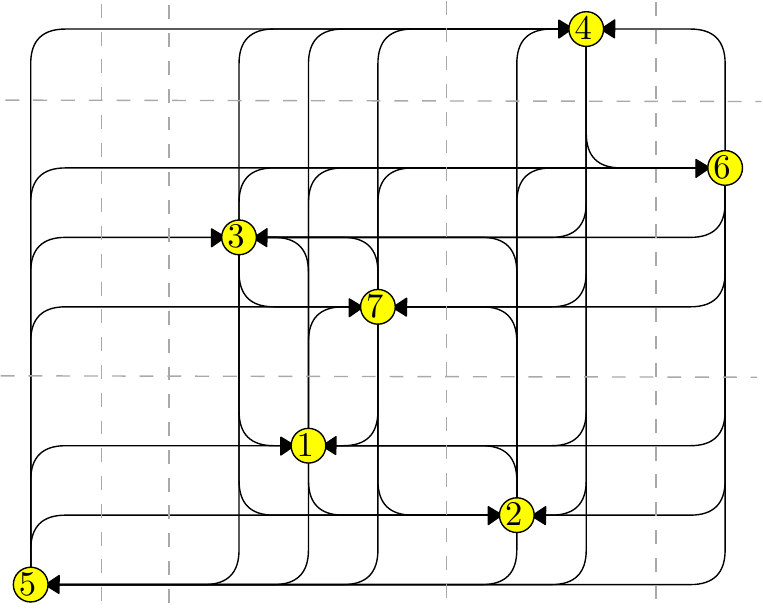} &
   \includegraphics[width=2.4cm]{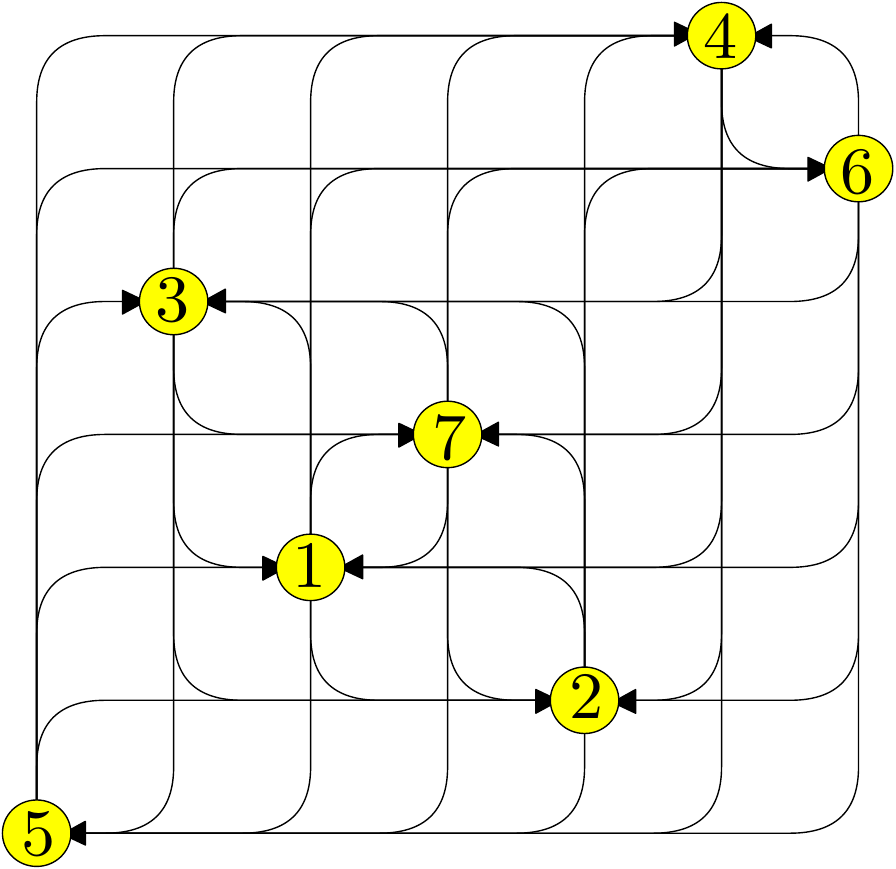} \\
   (a) & (b) & (c)\\
\end{tabular}
    \caption{(a) An L-drawing of $K_7$ on the $11 \times 10$ grid. (b) The L-drawing in (a) after the removal of a horizontal grid line. (c) A minimum-ink L-drawing of $K_7$.}\label{fi:k7}
    \end{center}
\end{figure}

Clearly, Lemma~\ref{le:complete-optimal} implies that any L-drawing of $K_n$ on the $n \times n$ grid is a minimum-ink drawing, since empty rows or columns never reduce the ink. However, when a complete graph $K_n$ is a subgraph of a larger graph, it might make sense to spread its vertices on a larger grid. In Lemma~\ref{le:complete-non-optimal} we hence study the cost, in terms of ink, of this operation.

\begin{lemma}\label{le:complete-non-optimal}
Any L-drawing of $K_n$ on the $(n+h) \times (n+k)$ grid uses $2n (n-1) + n(h+k)$ ink.
\end{lemma}
\begin{proof}
Let $\Gamma$ be an L-drawing of $K_n$ on the $(n+h) \times (n+k)$ grid. If $h>0$ consider any horizontal grid line $l$ not intersecting any vertex of $K_n$ and such that at least one vertex is above $l$ and at least one vertex is below $l$. For example, Fig.~\ref{fi:k7}(a) shows a drawing of $K_7$ on the $11 \times 10$ grid and a possible grid line $l$ in red. Denote by $p$ be the number of vertices above $l$ ($n-p$ is the number of vertices below $l$). Line $l$ is traversed by $p$ vertical segments of $\Gamma$ exiting the $p$ vertices above $l$ and entering the region below $l$. Also, $l$ is traversed by $n-p$ vertical segments exiting the $n-p$ vertices below $l$ and entering the region above $l$. Since vertices have exclusive $x$-coordinates, these $p+(n-p)=n$ vertical segments use distinct vertical grid lines; thus, removing line $l$ yields an L-drawing $\Gamma'$ on the $(n+h-1) \times (n+k)$ grid that saves $n$ ink (see Fig.~\ref{fi:k7}(b)).
Analogous compressions can be performed starting from vertical grid lines that do not intersect any vertex.
After $h+k$ compressions we produce an L-drawing of $K_n$ of minimum size (see Fig.~\ref{fi:k7}(c)) which, by Lemma~\ref{le:complete-optimal}, uses $2n(n-1)$ ink. It follows that the ink of the original drawing is  $2n(n-1) + n(h+k)$, hence the statement.\qed
\end{proof}

%%%%%%%%%
%%%%%%%%%
%%%%%%%%%
%%%%%%%%%
%%%%%%%%%
\section{Complexity of the MILD Problem}\label{se:complexity}

In order to show the NP-hardness of \textsc{MILD} we reduce the problem \textsc{Profile}, which is defined as follows.

\begin{quote}
{\bf Problem:} 
\begin{minipage}[t]{.90\linewidth} 
{\sc Profile}
\end{minipage}\par\vspace{2pt}
{\bf Instance:} 
\begin{minipage}[t]{.90\linewidth} 
A graph $G=(V,E)$ and an integer $k$.
\end{minipage}\par\vspace{2pt}
{\bf Question:} 
\begin{minipage}[t]{.90\linewidth} 
Does there exist an ordering $\pi$ for the vertices of $V$ such that 
\begin{equation}\label{eq:profile}
\sum\limits_{u \in V} \left( \pi(u)- \min\limits_{v \in N(u) \cup u}\pi(v)\right) \leq k\end{equation}
\end{minipage}
\end{quote}

It is folklore\footnotemark[1]\footnotetext[1]{Refer to~\cite{dps-sglp-02}. A formal proof of the equivalence of the two problems can be found in~\cite{gf-tvsnp-98}.} that \textsc{Profile} is equivalent to the NP-complete problem \textsc{SumCut}~\cite{dgpt-mp-91,g-tvsng-97,ly-pmpmg-94}, which is formally defined as follows.

\begin{quote}
{\bf Problem:} 
\begin{minipage}[t]{.90\linewidth} 
{\sc SumCut}
\end{minipage}\par\vspace{2pt}
{\bf Instance:} 
\begin{minipage}[t]{.90\linewidth} 
A graph $G=(V,E)$ and an integer $k$.
\end{minipage}\par\vspace{2pt}
{\bf Question:} 
\begin{minipage}[t]{.90\linewidth} 
Given an ordering $\pi$ for the vertices of $V$, denote by $\delta(i,\pi)$ the cardinality of $\{u: \exists (u,v) \in E: \pi(u) \leq i < \pi(v)\}$. Does there exist an ordering $\pi$ such that 
\begin{equation}\label{eq:sumcut}\sum\limits_{i = 1}^n \delta(i,\pi) \leq k\end{equation}
\end{minipage}
\end{quote}

Given an instance $I_p = \langle G=(V,E), k\rangle$ of \textsc{Profile}, we build an equivalent instance $I_m = \langle G'=(V',E'),k' \rangle$ of \textsc{MILD} as follows. Graph $G'$ contains two subgraphs $K^1$ and $K^2$, that are complete graphs on $p=\frac{5}{2}n^2 + \frac{9}{2}n + 1$ vertices, where $n=|V|$. Consider two arbitrary vertices $v_1$ and $v_2$ of $K^1$ and $K^2$, respectively. For each vertex $v \in V$ we add to $V'$ a vertex $u_v$ with (directed) edges $(u_v,v_1)$, $(u_v,v_2)$, and $(v_2,u_v)$. For each edge $e =(v,w) \in E$ we add to $E'$ edges $(u_v,u_w)$ and $(u_w,u_v)$. We set $k' = k + 4p(p-1) + \frac{3}{2}n^2 + \frac{9}{2}n$.

%%%%%%%
%     %  
%     %  
%%%%%%%
\begin{figure}[htb]
\begin{center}
   \includegraphics[page=1,width=7cm]{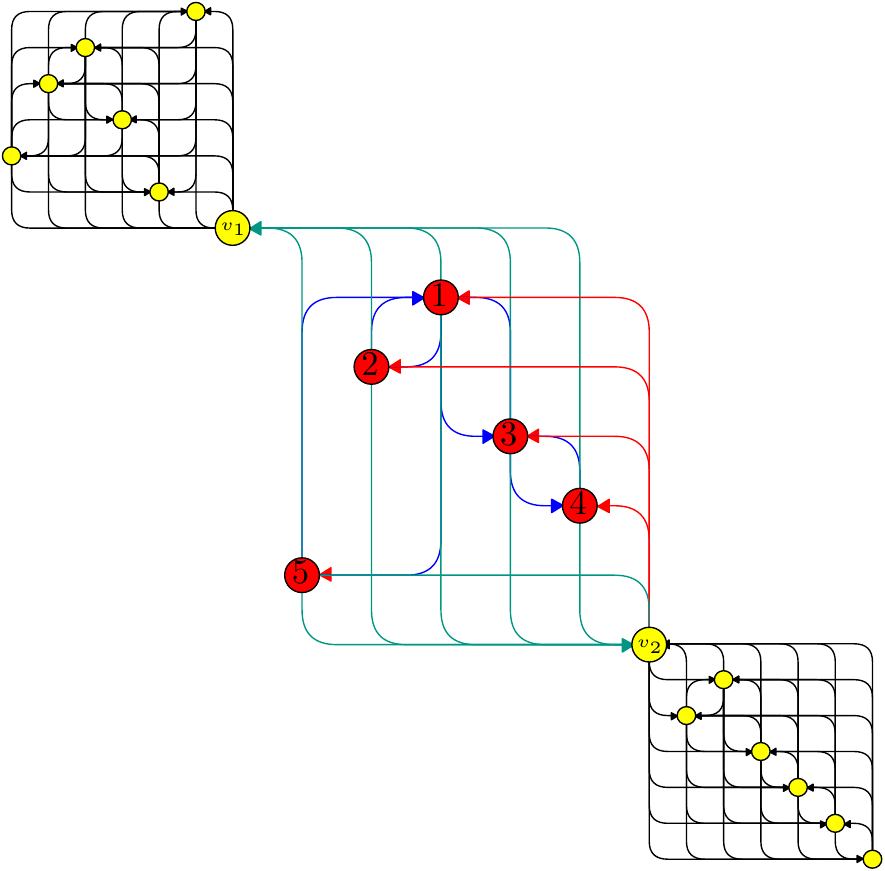}
    \caption{Instance $I_m$ of \textsc{Profile} with $p=7$ ($K^1$ and $K^2$ are drawn smaller for space reasons).}\label{fi:reduction}
    \end{center}
\end{figure}

\begin{lemma}\label{le:equivalence}
Instance $I_p$ admits a solution if and only if instance $I_m$ does.
\end{lemma}
\begin{proof}
Suppose instance $I_p= \langle G=(V,E), k\rangle$ of \textsc{Profile} admits a solution and let $\pi$ be the ordering of the vertices of $V$ such that $\sum_{u \in V} (\pi(u)- \min_{v \in N(u) \cup u}\pi(v)) \leq k$. We show that the corresponding instance $I_m= \langle G'=(V',E'),k' \rangle$ of \textsc{MILD} admits a solution. We draw $K^1$ and $K^2$ in such a way that each uses contiguous $x$- and $y$-coordinates and the bounding box of $K^1$ is above and on the left of the bounding box of $K^2$. In particular, we place $v_1$ in the bottom right corner of the bounding box of $K^1$ and $v_2$ in the top left corner of the bounding box of $K^2$. We insert between $K^1$ and $K^2$ the remaining part of the vertices in $V'$ so that their horizontal ordering corresponds to $\pi$ and their vertical ordering is arbitrary. See Fig.~\ref{fi:reduction} for an example.
We show that the ink is less than $k' = k + 4p(p-1) + \frac{3}{2}n^2 + \frac{9}{2}n$. In fact, the ink can be computed as a sum of: (i) the ink used inside the complete subgraphs $K^1$ and $K^2$ (black edges of Fig.~\ref{fi:reduction}), which by Lemma~\ref{le:complete-optimal} is $4p(p-1)$ in total; (ii) the ink used to connect for each $u \in V$ vertex $n_u$ to $v_1$ and $v_2$ (drawn in green in Fig.~\ref{fi:reduction}), which is $2n + (n+1)n$; (iii) the ink of the edges (drawn red in Fig.~\ref{fi:reduction}) that connect $v_2$ to $n_u$, for each $u \in V$, which is $n + \sum_{i=1}^n i$; (iv) the ink used for the edges among vertices $n_u$, with $u \in V$. The vertical ink of the latter contribution is already computed in (ii). The horizontal ink is exactly $\sum_{u \in V} (\pi(u)- \min_{v \in N(u) \cup u}\pi(v))$. Summing up the contributions (i)--(iii) we have $4p(p-1) + \frac{3}{2}n^2 + \frac{9}{2}n$. Since  $\sum_{u \in V} (\pi(u)- \min_{v \in N(u) \cup u}\pi(v)) \leq k$ the used ink is at most $k' = k + 4p(p-1) + \frac{3}{2}n^2 + \frac{9}{2}n$.
   
Conversely, suppose that instance $I_m$ admits an L-drawing using at most $k'$ ink. 
By Lemmas~\ref{le:complete-optimal} and~\ref{le:complete-non-optimal}, any L-drawing of $K^1$ or $K^2$ that does not use contiguous $x$- and $y$-coordinates uses at least $p = \frac{5}{2}n^2 + \frac{9}{2}n + 1$ ink more than an L-drawing that uses contiguous $x$- and $y$-coordinates. Observe that the value on the left side of equation~\ref{eq:profile} is bounded by $n^2$, where $n = |V|$. Hence, we can assume $k \leq n^2$ in any non-trivial \textsc{Profile} instance. It follows that the additional ink that would be needed to insert grid lines in the drawings of $K^1$ or $K^2$ is at least $p = \frac{5}{2}n^2 + \frac{9}{2}n + 1 > k + \frac{3}{2}n^2 + \frac{9}{2}n$. This ensures that in any L-drawing that uses at most $k'$ ink, $K^1$ and $K^2$ use contiguous $x$- and $y$-coordinates, and vertices $n_u$, for each $u \in V$, are inserted between the bounding boxes of $K^1$ and $K^2$, both in the horizontal and in the vertical order. Hence, by Lemma~\ref{le:complete-optimal}, the total contribution of these two subgraphs is $4p(p-1)$. 
Also, $v_1$ lies on the bottom-right corner of $K_1$ and $v_2$ on the top-left corner of $K_2$, as they are the only vertices of $K_1$ and $K_2$ that are connected to vertices $n_u$, with $u \in V$.
This implies that, for every horizontal and vertical order of vertices $n_u$, the cost of the green edges in Fig.~\ref{fi:reduction} is $2n + (n+1)n$, and the cost of the red edges is $n + \sum_{i=1}^n i$. Finally, for the blue edges, the vertical contribution is already covered by the green edges and the horizontal contribution is less of equal $k$. Hence, the horizontal order of vertices $n_u$ yields a solution for \textsc{Profile}.
This concludes the proof. \qed
\end{proof}

\begin{theorem}\label{th:np-complete}
\textsc{MILD} is NP-complete.
\end{theorem}
\begin{proof}
\textsc{MILD} is trivially in NP by non-deterministically trying all permutations $\pi_x$ and $\pi_y$ of the vertices of the graph and computing the ink of $\Gamma(\pi_x,\pi_y)$. Given an instance $I_p$ of \textsc{Profile}, the corresponding instance $I_m$ of \textsc{MILD} can be built in polynomial time, and Lemma~\ref{le:equivalence} ensures that the two instances are equivalent.\qed 
\end{proof}

%%% end of remove

%%%%%%%%%
%%%%%%%%%
%%%%%%%%%
%%%%%%%%%
%%%%%%%%%
\section{A Polynomial On-line Algorithm}\label{se:online}

Motivated by the NP-completeness result in Theorem~\ref{th:np-complete}, we seek in this section for an efficient heuristics to construct L-drawings of graphs with reduced ink. In particular, we study the setting in which the drawing is constructed incrementally by adding one vertex at a time to a previously computed drawing; the goal is then to add the new vertex (with all its incident edges) using the minimum additional ink, where the only operation that is allowed on the previous drawing is to insert a row and a column (the cost of elongating the edges traversing the inserted row/column has hence to be taken into account, as well). We prove in Theorem~\ref{th:incremental-optimal} that there exists a polynomial-time algorithm, called \texttt{OptAddVertex}, to place the given vertex in the given L-drawing while minimizing the additional ink of the resulting L-drawing with respect to the given one.

We remark that, besides providing a heuristics for the general problem, this incremental approach fits in the framework of streamed graph drawing, in which the graph to be drawn is too large to be stored in the memory and hence comes in the form of a streaming of its elements (vertices, edges, components) that have to be placed in the drawing without a prior knowledge of the elements that are yet to come.

Since, by Lemma~\ref{le:separazione}, the horizontal and vertical coordinates of an L-drawing can be computed independently, we describe Algorithm \texttt{OptAddVertex} by only focusing on how to compute the optimal $x$-coordinate of the new vertex, adding a column.

Let $G=(V,E)$ be an $n$-vertex directed graph and let $\Gamma$ be an L-drawing of it. We assume that the vertices in $V$ have $x$-coordinates in $\{1,2,\dots,n\}$. Vertex $v$ has to be added to the drawing, with its (possibly empty) set of outgoing edges $\{(v,u_1)$, $(v,u_2)$, \dots, $(v,u_h)\}$ towards vertices of $V$ and its (possibly empty) set of incoming edges $\{(w_1,v)$, $(w_2,v)$, \dots, $(w_k,v)\}$ from vertices of $V$. 

Algorithm \texttt{OptAddVertex} computes the additional ink needed to insert a vertical grid line $l_v$ for $v$ in each one of the possible $n+1$ positions $\{1,2, \dots,n+1\}$, where if $l_v$ is inserted in position $i$, all vertices of $V$ with $x$-coordinate greater or equal than $i$ have to be shifted one unit to the right (hence, $i=1$ and $i=n+1$ correspond to adding a column to the left and to the right of the drawing, respectively). 

We define three integer functions, that we call $StretchInk_x$, $IncomingInk_x$, and $OutgoingInk_x$, in the domain $\{1,2, \dots,n+1\}$ as follows. $StretchInk_x(i)$ is the cost of inserting $l_v$ in position $i$. This cost is due to the fact that the length of all horizontal segments traversed by $l_v$ is incremented by one. $IncomingInk_x(i)$ is the cost, in terms of horizontal ink, of routing the edges $\{(w_1,v)$, $(w_2,v)$, \dots, $(w_k,v)\}$ entering $v$ when $v$ is placed in position $i$. Observe that all these edges will enter $v$ on a horizontal grid line $l_h$, which is exclusive of $v$. Hence, the value of $IncomingInk_x(i)$ is the range of the $x$-coordinates of vertices $\{w_1,, w_2, \dots, w_k\} \cup \{v\}$ after the insertion of $v$ in position $i$. 
The computation of function $OutgoingInk_x$ is more complex. Each outgoing edge $(v,u_j)$, $j=1,\dots,h$, of $v$ has a vertical segment (which does not contribute to function $OutgoingInk_x$) and an horizontal segment entering $u_j$ at its $y$-coordinate $y_{u_j}$. However, $u_j$ may have already horizontal segments entering it at $y$-coordinate $y_{u_j}$. Let $W_j$ and $E_j$ be the minimum and the maximum $x$-coordinate that are used by some horizontal segments at coordinate $y=y_{u_j}$ (if there is no horizontal segment with $y=y_{u_j}$ we set $W_j=E_j=x_{u_j}$). The contribution of edge $(v,u_j)$ to $OutgoingInk_x(i)$ is zero if $W_j \leq i \leq E_j$ and $\min(|i-W_j|,|i-E_j|)$, otherwise.   

Finally, we insert $v$ in a position corresponding to a minimum of function $AddInk_x$ defined as $AddInk_x = StretchInk_x + IncomingInk_x + OutgoingInk_x$.

The heuristic \texttt{IncrementaLDraw} for producing L-drawings of directed graphs works as follows. First, we order the vertices of the graph in such a way that, for any $1 \leq j \leq n$, the subgraph induced by the first $j$ vertices is connected. In particular, we consider the vertices in a BFS order. Second, we assign to the first vertex coordinates $(1,1)$ and add a vertex at a time in the given order using Algorithm \texttt{OptAddVertex}.

We say that a permutation $\pi_1$ of the first $n$ positive integers \emph{extends} a permutation $\pi_2$ of the first $n-1$ positive integers if $\pi_2$ can be obtained from $\pi_1$ by removing element~$n$. 

\begin{theorem}\label{th:incremental-optimal}
Given a directed graph $G$, a vertex $v \in G$, and an L-drawing $\Gamma'(\pi'_x,\pi'_y)$ of the subgraph $G' = G \setminus v$, algorithm \texttt{OptAddVertex} constructs in linear time an L-drawing $\Gamma^*(\pi^*_x,\pi^*_y)$ of minimum ink among all L-drawings $\Gamma(\pi_x,\pi_y)$ of $G$ such that $\pi_x$ extends $\pi'_x$ and $\pi_y$ extends $\pi'_y$.
\end{theorem}
\begin{proof}
Suppose by contradiction that there exists an L-drawing $\Gamma^\circ(\pi^\circ_x,\pi^\circ_y)$ that uses less ink than $\Gamma^*(\pi^*_x,\pi^*_y)$ and such that $\pi^\circ_x$ extends $\pi'_x$ and $\pi^\circ_y$ extends $\pi'_y$. Without loss of generality suppose that $ink_x(\Gamma^\circ) < ink_x(\Gamma^*)$. By removing $v$ we obtain again $\Gamma'(\pi'_x,\pi'_y)$ and we save $AddInk_x(x^\circ_v)$ ink, where $x^\circ_v$ is the $x$-coordinate of $v$ in $\Gamma^\circ$. Since $ink(\Gamma^*) = ink(\Gamma')+AddInk_x(x^*_v)$, where $(x^*_v)$ is the $x$-coordinate of $v$ in $\Gamma^*$, we have that $AddInk_x(x^\circ_v)< AddInk_x(x^*_v)$, contradicting the hypothesis that $\Gamma^*$ is obtained by inserting $v$ in a minimum of function $AddInk_x$. 

We now show how to produce a linear-time implementation of Algorithm \linebreak \texttt{OptAddVertex}. 
We give some hint about how functions $StretchInk_x$, $Inco\-mingInk_x$, and $Out\-going\-Ink_x$ can be computed in linear time.

Maintain a data structure that contains, for each vertex $v_i$, $i=1, \dots, n$, the $x$-coordinates of its leftmost bend $W_i$ and rightmost bend $E_i$. For the computation of $StretchInk_x$ proceed as follows: (i) produce a single ordering of the $W_i$ and $E_i$ based on their $x$-coordinates (a bucket sort would take linear time); (ii) examine the $E_i$ and $W_i$ in order. Let $x$ be the $x$-coordinate of the current element. If the current element is the first considered with coordinate $x$ then initialize $StretchInk_x(x)=StretchInk_x(x-1)$ (where it is assumed $StretchInk_x(0)=0$). If the current element is a $W_i$, add $1$ to $StretchInk_x(x)$, otherwise subtract one from $StretchInk_x(x)$.
  
Function $IncomingInk_x$ is easy to compute in linear time by following the description in Section~\ref{se:online}. 

A linear-time implementation of Function $OutgoingInk_x$ is more complex, but can be done with a similar approach to that of $StretchInk_x$. Namely, let $v$ be the current inserted vertex and let $u_j$, $j=1,\dots, h$, be the vertices that have an incoming edge $(v,u_j)$. For each $j=1,\dots,h$ we have to compute the contribution to function $OutgoingInk_x$, which is zero if $x$ falls in between $W_j$ and $E_j$, and increases linearly with the distance of $x$ from the nearest between $W_j$ and $E_j$. Based on their $x$-coordinates, separately order the $W_j$ and the $E_j$, $j=1, \dots, h$. Now, with a sweep from left to right of the $E_j$'s add to $OutgoingInk_x$ the contribution of vertices $u_j$ when $x > E_j$ and with a second sweep from right to left of the $W_j$'s add to $OutgoingInk_x$ the contribution of vertices $u_j$ when $x < W_j$.\qed
\end{proof}

\subsection{Running Times of Algorithm \texttt{IncrementaLDraw}}\label{apx:A4}

Our JavaScript na\"{\i}ve implementation of Algorithm \texttt{IncrementaLDraw} is not optimized for efficiency and has an $O(n^3)$ time-complexity because our implementation of Algorithm \texttt{OptAddVertex} has $O(n^2)$ time-complexity. Figs.~\ref{fi:times-a} and~\ref{fi:times-b} shows the average times and standard deviations over $100$ runs of \texttt{IncrementaLDraw} on the second graph suite (each run uses a different ordering of the vertices obtained by starting the BFS from a random vertex and by shuffling the adjacency lists of the vertices). In particular, Fig.~\ref{fi:times-a} shows the curves for the different densities of the edges ($10\%$, $20\%$, $30\%$, and $70\%$ of the maximum possible number of edges), while Fig.~\ref{fi:times-b} shows the averages and standard deviations over all densities.

\remove{
%%%%%%%
%     %  
%     %  
%%%%%%%
\begin{figure}[htb]
\begin{center}
\begin{tabular}{c  c  }
   \includegraphics[width=6cm]{figures/timing_01_density} & 
   \includegraphics[width=6cm]{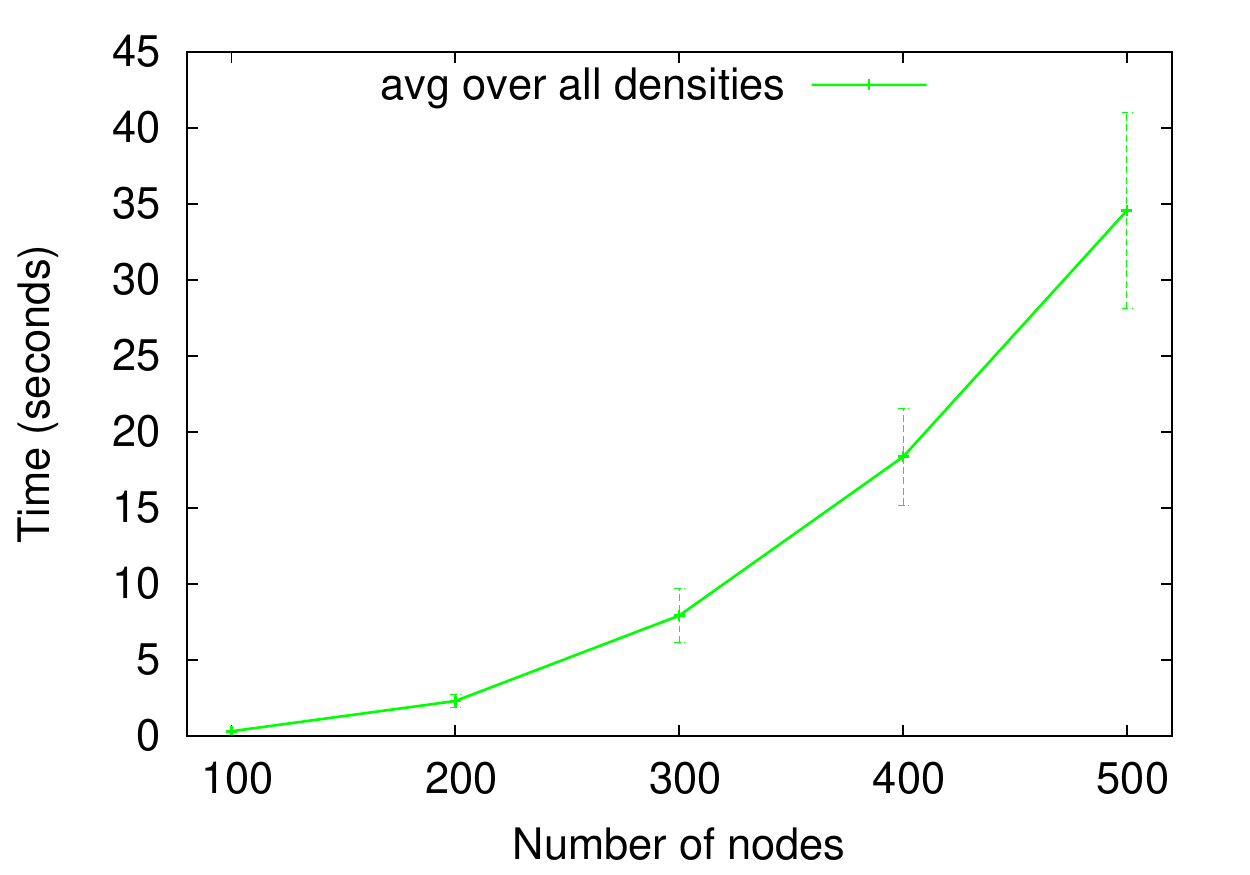} \\
   (a) & (b) \\
\end{tabular}
    \caption{Running times for Algorithm \texttt{IncrementaLDraw} on the graphs of the second graph suite. (a) Average times and standard deviations for different densities of the edges. (b) Average times and standard deviations over all the second graph suite.}\label{fi:times}
    \end{center}
\end{figure}
} 

%%%%%%%
%     %  
%     %  
%%%%%%%
\begin{figure}[tb!]
\centering
   \subfigure[]{\includegraphics[width=8cm]{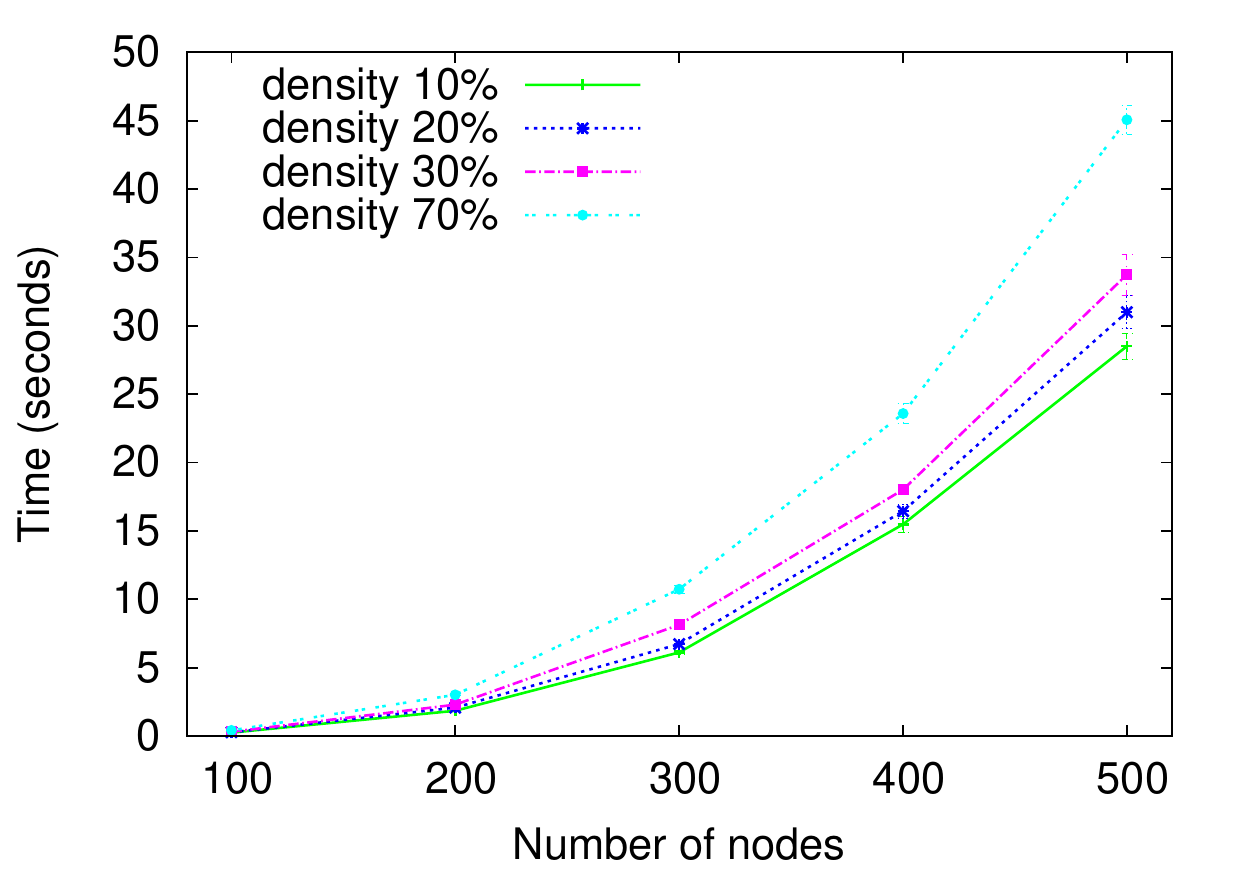}\label{fi:times-a}}
   \subfigure[]{\includegraphics[width=8cm]{figures/timing_all_densities}\label{fi:times-b}}
    \caption{(a) Running times for Algorithm \texttt{IncrementaLDraw} on the graphs of the second graph suite. Average times and standard deviations for different densities of the edges. (b) Average running times and standard deviations for Algorithm \texttt{IncrementaLDraw} over all the graphs of the second graph suite.
    }
\end{figure}

%%%%%%%
%     %  
%     %  
%%%%%%%

%%%%%%%%%
%%%%%%%%%
%%%%%%%%%
%%%%%%%%%
%%%%%%%%%
\section{Experimental Evaluation}\label{se:experiments}

We implemented Algorithm \texttt{OptAddVertex} and the heuristics \texttt{IncrementaLDraw}, and performed an extensive testing to evaluate the quality of the obtained L-drawings. We compared the performances of our heuristics with the optimum ink, the \texttt{OOD} algorithm of DAGView~\cite{kt-davlg-12}, and random placements.

\subsection{An ILP Formulation}

In order to compare the heuristic approach with the optimal solution we formulated the problem of finding an L-drawing with minimum ink as an ILP problem. 
Figs.~\ref{fi:examples-2} provides some examples of L-drawings computed by Algorithm \texttt{Incre\-men\-taLDraw}. The drawing produced by Algorithm \texttt{IncrementaLDraw} is compared with the minimum ink drawing obtained via the IPL formulation.

%%%%%%%
%     %  
%     %  
%%%%%%%
\begin{figure}[tb]
\begin{center}
\begin{tabular}{c @{\hspace{3em}} c  }
   \includegraphics[width=5.5cm]{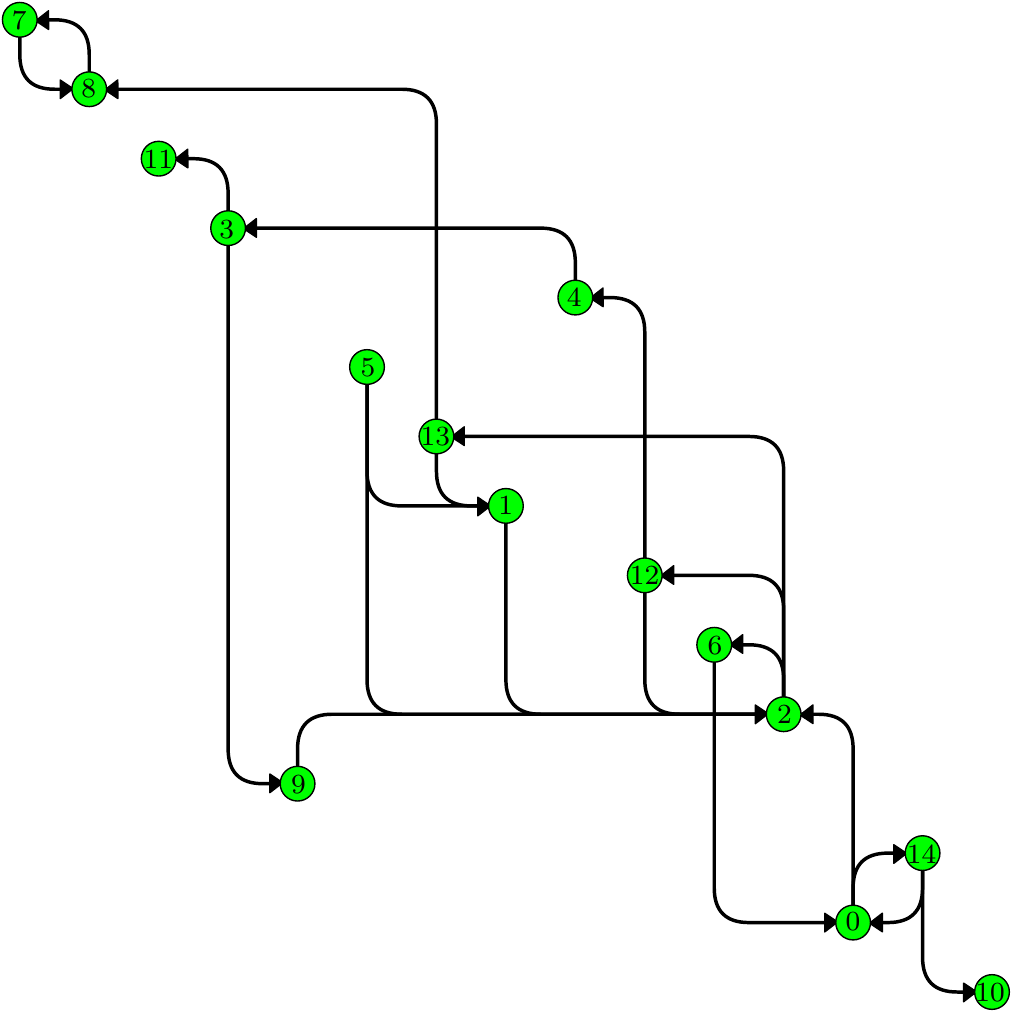} & 
   \includegraphics[width=5.5cm]{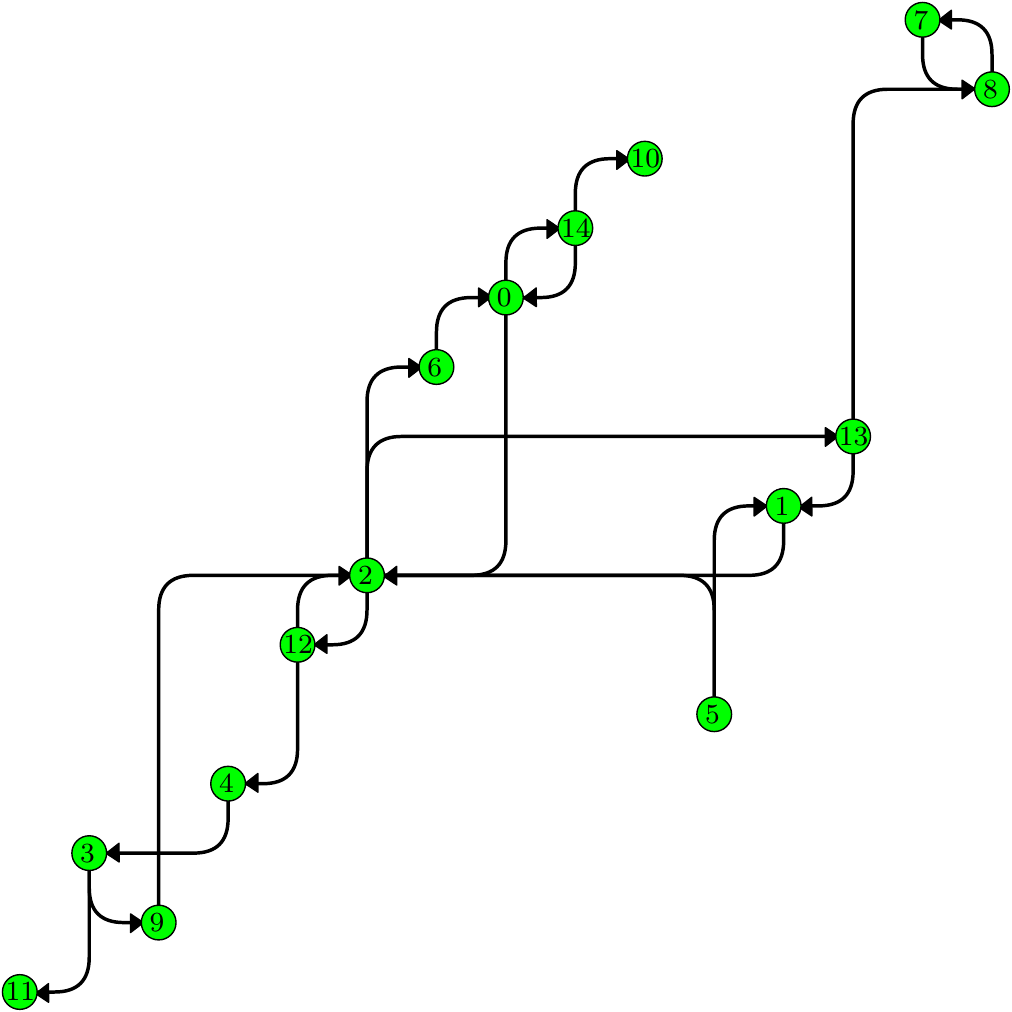} \\
   (a) & (b) \\
\end{tabular}
    \caption{(a) An L-drawing of a random graph with $15$ vertices and $21$ edges ($10\%$ of the maximum possible) drawn with Algorithm~\texttt{IncrementaLDraw}. Ink consumption is $84$. (b) The minimum ink L-drawing of the same graph. Ink consumption is $66$.}\label{fi:examples-2}
    \end{center}
\end{figure}

%%%%%%%
%     %  
%     %  
%%%%%%%
\remove{
\begin{figure}[tb]
\begin{center}
   \includegraphics[width=8cm]{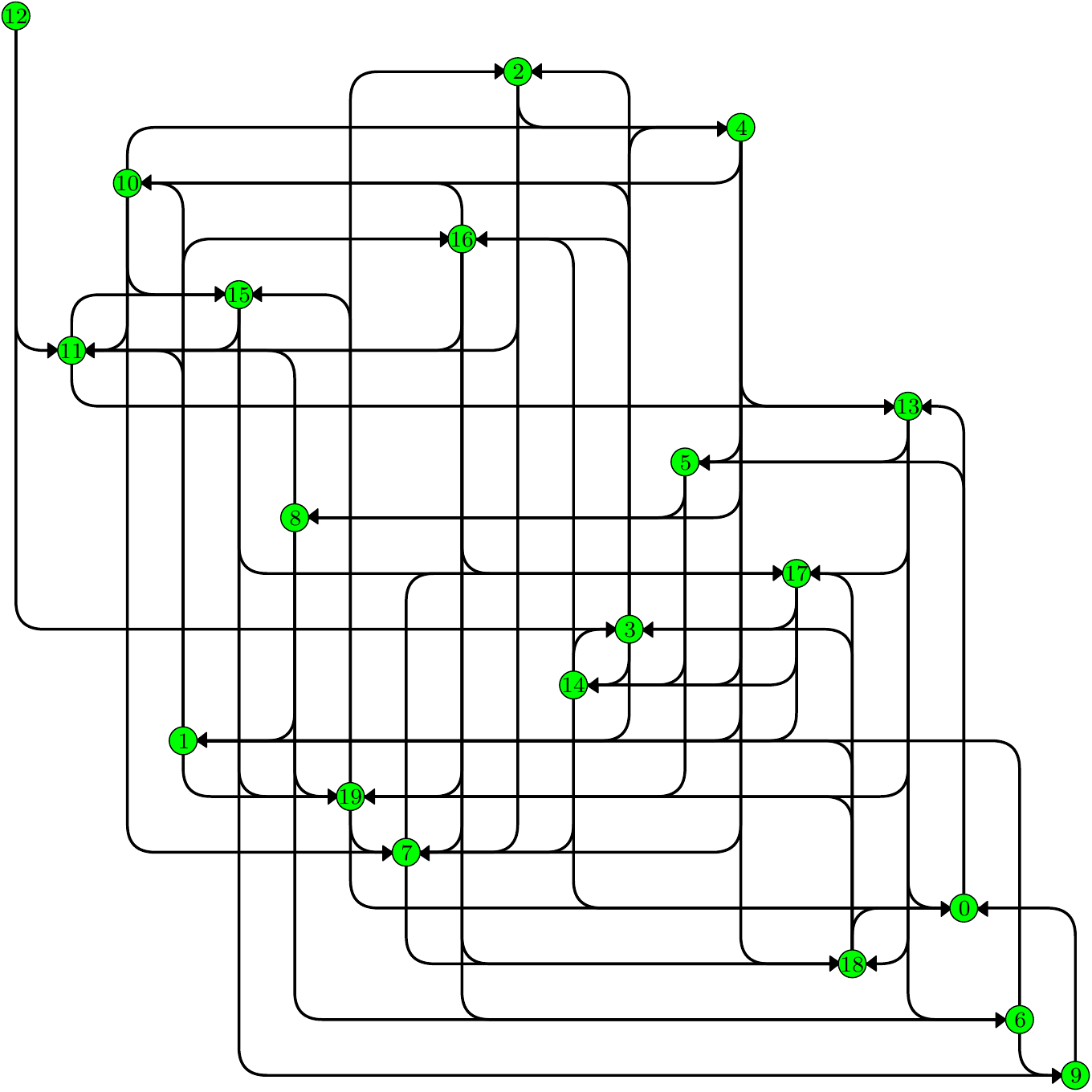}
    \caption{An L-drawing of a random graph with $20$ vertices and $76$ edges ($20\%$ of the maximum possible) drawn with Algorithm \texttt{IncrementaLDraw}.}\label{fi:examples-3}
    \end{center}
\end{figure}
}

Given an n-vertex graph $G=(V,E)$, in the following we describe only the part to compute its $x$-coordinates (the computation of $y$-coordinates is analogous). By definition, the amount $ink_x$ of a drawing $\Gamma$ of $G$ is obtained by summing up all the horizontal segments of the drawing. Since each $y$-coordinate is exclusively used for one vertex, there are $n$ (possibly null, if there exists a vertex with no incoming edges) horizontal segments in $\Gamma$. The horizontal segment $s_i$ that includes $v_i$, $i=1,\dots,n$ extends from the leftmost to the rightmost bends of the edges entering $v_i$. We call $W_i$ and $E_i$ the $x$-coordinates of the endpoints of $s_i$.    
Variables:
{\small
\begin{eqnarray*}
\forall i,j =1,\dots,n: & x_{ij}= \begin{cases}
1\mbox{ if vertex $v_i$ has $x$-coordinate $j$},\\
0\mbox{ otherwise}
\end{cases}\\
\forall i=1,\dots, n:& E_i, W_i \mbox{~~~~~~~(rightmost and leftmost endpoints of $s_i$)}\\
\end{eqnarray*}
}%
Variables $x_{ij}$ are binary, while $E_i$, $W_i$ are integers.  
To simplify the description of the constraints we denote by $x_i$ the $x$-coordinate of vertex $v_i$, that is $x_i = \sum_{j=1}^{n} x_{i,j}\cdot j$. 
%\\ \newline
{Constraints: }
{\small
\begin{eqnarray*}
\forall i,\  \sum_{j=1}^n x_{ij}=1 & \mbox{(each vertex has a unique $x$-coordinate)}\\
\forall j,\  \sum_{i=1}^n x_{ij}\le1 & \mbox{(each column contains at most one vertex)}\\
\forall i,\ E_i \ge x_i & \mbox{(the rightmost endpoint of $s_i$ does not lie to the left of $v_i$)}\\
\forall i,\ W_i \le x_i & \mbox{(the leftmost endpoint of $s_i$ does not lie to the right of $v_i$)}\\
\forall (v_i,v_j) \in E, \ E_j \ge x_i & \mbox{~~~(the rightmost endpoint of $s_j$ does not lie to the right of $v_i$)}\\
\forall (v_i,v_j) \in E, \ W_j \le x_i & \mbox{(the leftmost endpoint of $s_j$ does not lie to the right of $v_i$)}
\end{eqnarray*}
} % end \small
The objective function is: $\min \sum_{i=1}^n (E_i - W_i)$. 

To compute minimum-ink L-drawings we used Gurobi Optimizer ver. 6.0.4~\cite{gurobi} on a\remove{ Dell PowerEdge 2900III} Dual Xeon X5460 Quad Core 3.16GHz 48GB RAM.

%
%%%
%%%%%
%%%
%
\subsection{Random Generation of the Graphs Suites}

We generated uniformly at random two graph suites of dense, weakly connected, directed graphs. The first graph suite is meant to compare the performances of Algorithm \texttt{IncrementaLDraw} with respect to the optimum. For each number of vertices $n$ in $\{5, 10, 15\}$ and for each percentage $p$ in $\{10, 20, 30, 70\}$ we generated ten graphs whose number of edges $m$ is $p\%$ with respect to the maximum possible number of edges, that is, $m$ = $\lfloor n(n-1)p/100 \rfloor$. In particular, we used the procedure \texttt{gnm\_random\_graph} of the NetworkX 1.7 library~\cite{networkx}, discarding graphs that were not connected. 

The second graph suite is meant to compare \texttt{IncrementaLDraw} with a random placement of the vertices and is generated with the same procedure and edge percentages of the first suite, but vertices range in $\{100,200,300,400,500\}$.

%%%%%%%
%     %  
%     %  
%%%%%%%

%%%%%%%
%     %  
%     %  
%%%%%%%
\begin{figure}[htb]
\begin{center}
\begin{tabular}{c  c  }
   \includegraphics[width=6cm]{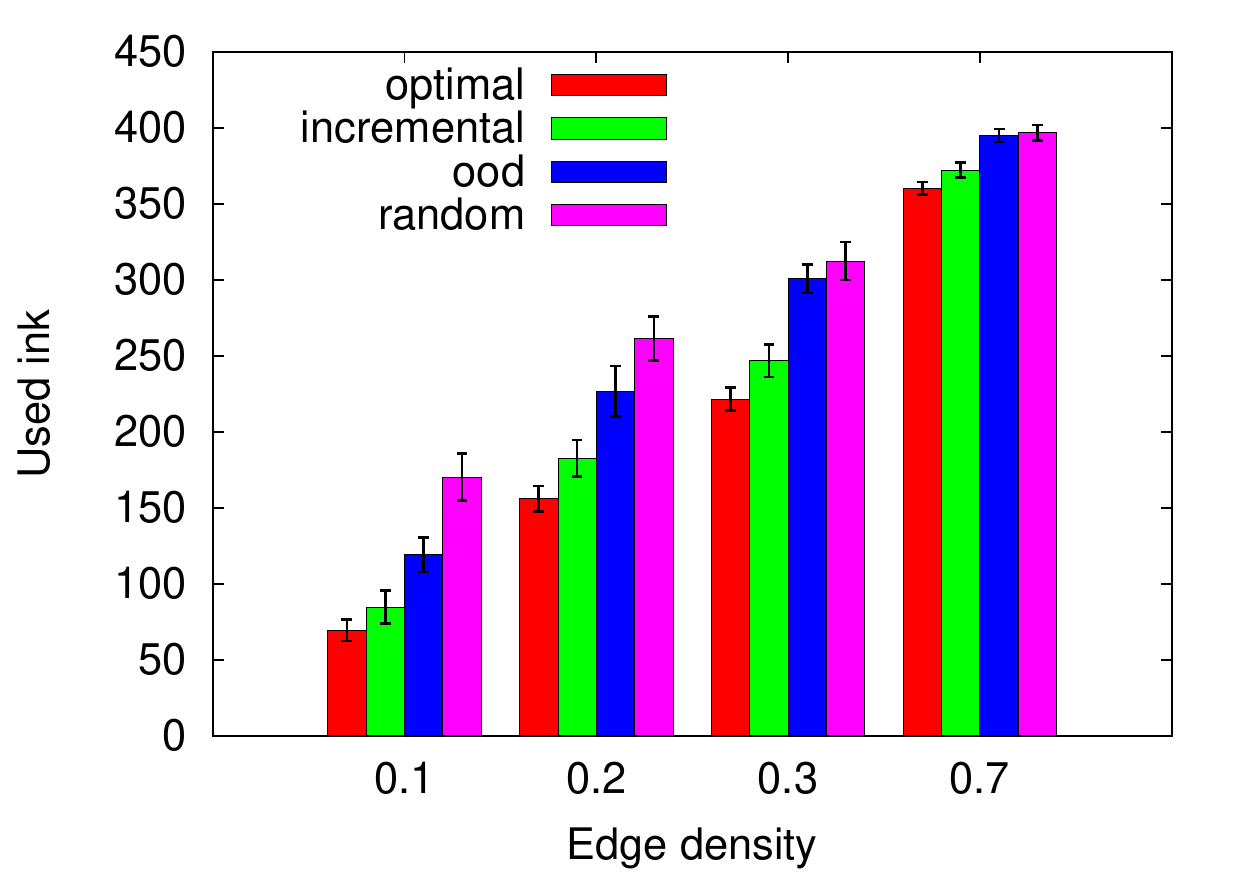} & 
   \includegraphics[width=6cm]{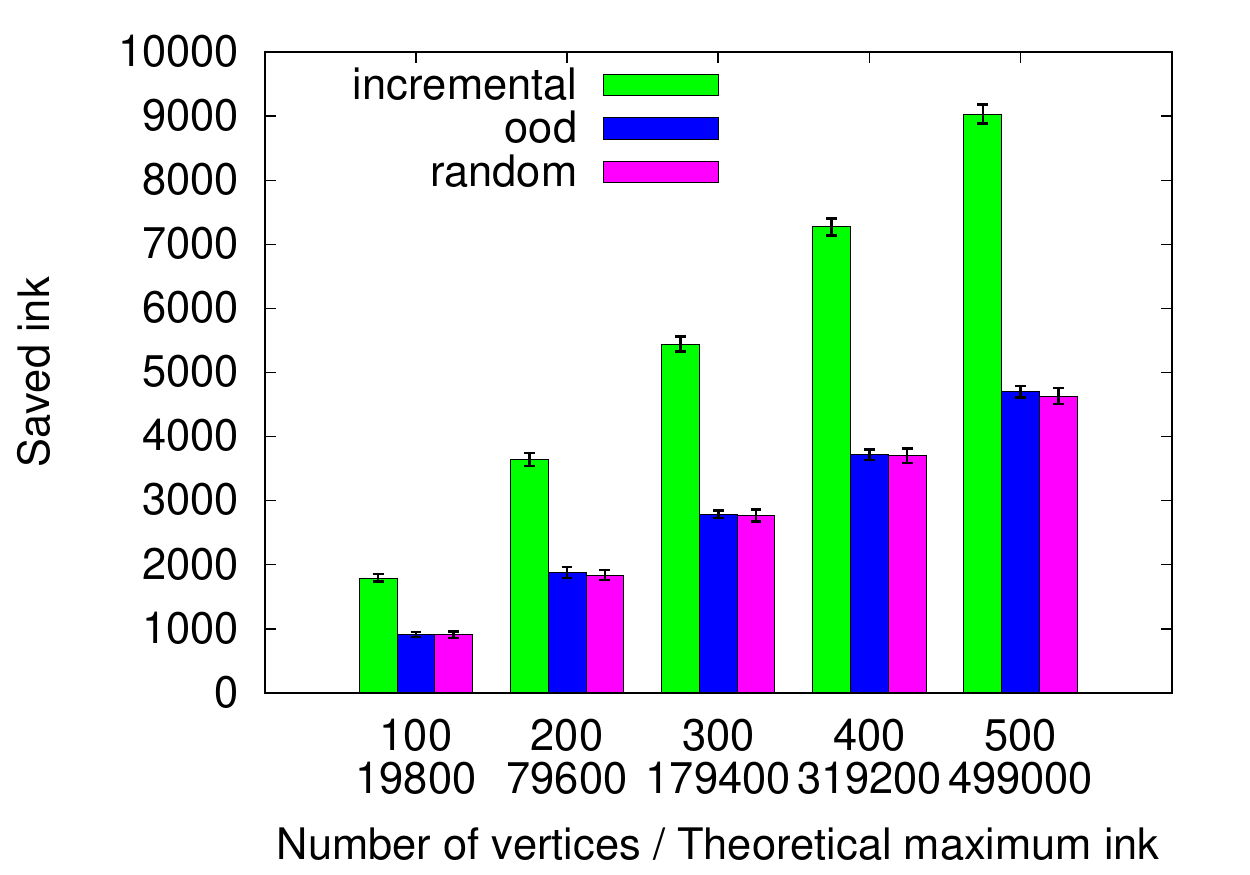} \\
   (a) & (b) \\
\end{tabular}
    \caption{(a) Ink consumption by varying density (the size of the graphs is fixed at $15$ vertices). (b) The difference between the theoretical maximum and the actual ink used by incremental, \texttt{ODD}, and random placement, for the second test-suite (graphs with $30\%$ of maximum possible edges).}\label{fi:results-2}
    \end{center}
\end{figure}
%
%%%
%%%%%
%%%
%
\subsection{Results of the Experiments}

\begin{figure}[htb]
\begin{center}
\begin{tabular}{c  c  }
   \includegraphics[width=6cm]{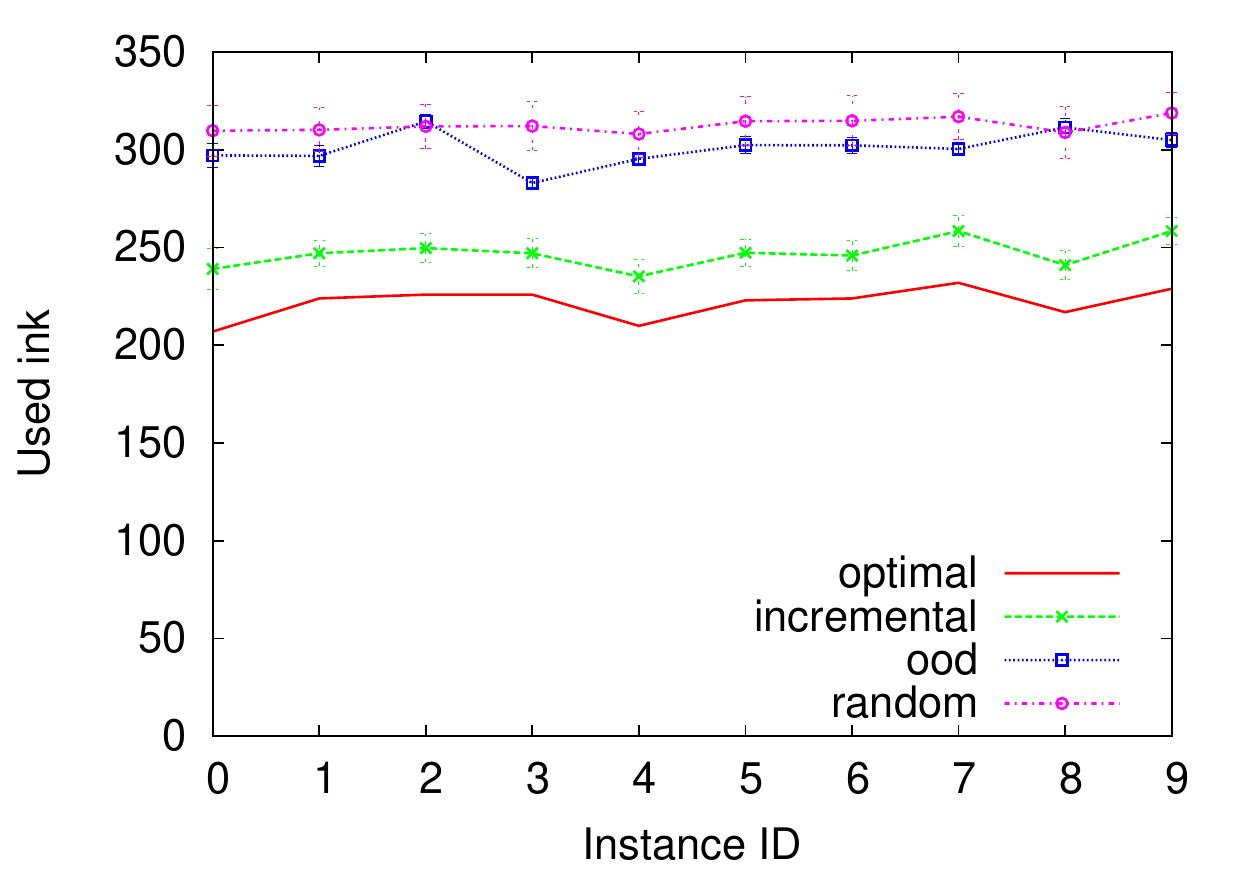} & 
   \includegraphics[width=6cm]{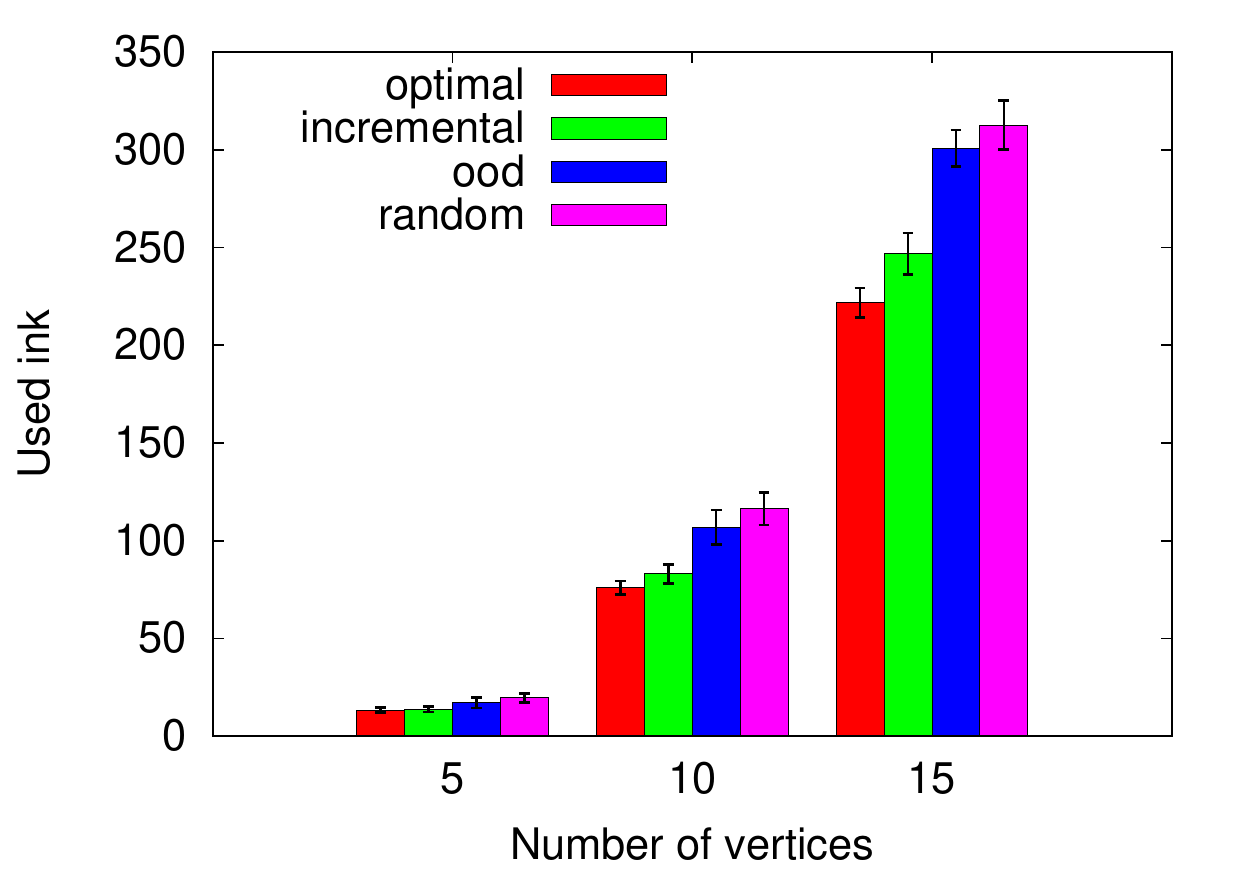} \\
   (a) & (b) \\
\end{tabular}
    \caption{(a) Ink used for the drawings of the ten graphs of the first graph suite with 15 vertices and 63 edges (corresponding to $30\%$ of the maximum possible). Optimal, incremental, ODD, and random placement are compared. For the latter three the average and the standard deviation over 100 runs is shown. (b) Ink consumption varying the size of the graphs \mbox{(fixing at $30\%$ edge density).}}\label{fi:results-1}
    \end{center}
\end{figure}

The results of the experiments are shown in Figs.~\ref{fi:results-1} and~\ref{fi:results-2}. Fig.~\ref{fi:results-1}(a) is devoted to the ten graphs with $15$ vertices and $63$ edges ($30\%$ of the maximum possible) of the first graph suite. On the $x$-axis the ten graphs are reported. 
The curves represent: (i) the ink used by the optimal algorithm; (ii) the average and the standard deviation of the ink used by Algorithm \texttt{IncrementaLDraw} over $100$ runs, each using a different BFS ordering obtained by starting from a random initial vertex and by shuffling the adjacency lists of the vertices; (iii) the average and the standard deviation of the ink used by Algorithm \texttt{OOD} over $100$ runs, each obtained from DAGView~\cite{kt-davlg-12} by shuffling the adjacency lists of the vertices; and (iv) the average and the standard deviation of the ink used by $100$ random placements of the vertices. 

From Fig.~\ref{fi:results-1}(a) it is apparent that the performances of \texttt{IncrementaLDraw} are always largely better than those of \texttt{OOD} and random placements, and not rarely are close to the optimum. Although this result could be anticipated (\texttt{OOD} was not conceived to reduce ink), we were surprised to note that, even with very small graphs and relatively many runs, the worst case for \texttt{IncrementaLDraw} is always comparable with the best case for \texttt{OOD} and significantly better than the best case of random placement. We found the same pattern in all plots obtained by changing densities and sizes.
 
Fig.~\ref{fi:results-1}(b) shows how the size impacts on ink, focusing on $30\%$ density graphs of the first graph suite. All the points are obtained by averaging ten values (for example, each bar for $15$ vertices of Fig.~\ref{fi:results-1}(b) is obtained by averaging the ten corresponding values of Fig.~\ref{fi:results-1}(a)). 
Fig.~\ref{fi:results-2}(b) further deepens this analysis showing how much ink each algorithm saves with respect to the maximum theoretical upper bound of $2n \times (n-1)$ for the second graph suite.
We observe that, when increasing the number of vertices, both the number of edges and the consumption of ink increase quadratically. At the same time, the ink saved by \texttt{IncrementaLDraw} with respect to \texttt{OOD} and random placement increases linearly.

Fig.~\ref{fi:results-2}(a) shows how density impacts on ink, focusing on graphs of $15$ vertices. Again, each point is the average of ten points obtained for ten different graphs (e.g., the values for density $30\%$ are obtained by averaging the ten values of Fig.~\ref{fi:results-1}(a)). For denser graphs, the difference among the alternative algorithms seems to reduce. This could be predicted as Lemma~\ref{le:complete-optimal} ensures that for any vertex order of a $K_n$ uses the same ink.

Overall, the experiments show that the ink consumption of \texttt{IncrementaLDraw} are closer to the optimum than to those of alternative algorithms and that the heuristics offers a good compromise between effectiveness and running times, even with a na\"{\i}ve implementation of Algorithm \texttt{OptAddVertex}. 

%%%%%%%%%
%%%%%%%%%
%%%%%%%%%
%%%%%%%%%
%%%%%%%%%

\section{Conclusions and Open Problems}\label{se:conclusions}

We introduced L-drawings, a novel paradigm for representing directed graphs. We investigated the problem of producing drawings with minimum ink, which turned out to be NP-complete. Our heuristics, however, proved to produce near-optimal solutions. 

Several problems remain open: (i) How much area and ink could be saved if vertices were allowed to share horizontal or vertical grid lines, provided that the drawing is still unambiguous?
(ii) Does there exist an ordering of the vertices such that \texttt{IncrementaLDraw} produces a minimum-ink drawing? (iii) Problem \textsc{Profile}, which we reduced to show the NP-hardness of \textsc{MILD}, is linear-time solvable for trees~\cite{dgpt-mp-91} and for square grids~\cite{dpps-ctslm-00}; what is the complexity of computing minimum-ink L-drawings for these families of graphs? 

Finally, although in~\cite{dmpt-hvdgu-14} it is shown that overloaded orthogonal drawings are superior to matrix representations under several respects, it would be interesting to contrast both these representations with L-drawings in an extensive user study.

\bibliographystyle{splncs03}
\bibliography{bibliography}

\end{document}